\documentclass[a4paper,11pt]{article}

\usepackage{fullpage}
\usepackage{times}
\usepackage{soul}
\usepackage{url}
\usepackage[utf8]{inputenc}
\usepackage[small]{caption}
\usepackage{graphicx}
\usepackage{amsmath}
\usepackage{booktabs}

\usepackage[ruled,vlined]{algorithm2e}
\SetArgSty{textrm} 

\usepackage{enumitem}

\usepackage{libertine}

\usepackage{amssymb,amsmath,amsfonts,amstext,amsthm}

\usepackage[breaklinks=true]{hyperref}
\usepackage[svgnames]{xcolor}
\usepackage[capitalise,nameinlink]{cleveref}
\hypersetup{colorlinks={true},linkcolor={DarkBlue},citecolor=[named]{DarkGreen}}

\usepackage{tikz}  
\usetikzlibrary{arrows}
\usetikzlibrary{patterns,snakes}
\usetikzlibrary{decorations.shapes}
\tikzstyle{overbrace text style}=[font=\tiny, above, pos=.5, yshift=5pt]
\tikzstyle{overbrace style}=[decorate,decoration={brace,raise=5pt,amplitude=3pt}]
\usetikzlibrary{shapes.geometric}
\usetikzlibrary{shapes,positioning}
\usetikzlibrary{calc,positioning}
\usetikzlibrary{shapes.multipart}

\definecolor{cadmiumgreen}{rgb}{0.0, 0.42, 0.24}

\usepackage{bbm}

\newtheorem{theorem}{Theorem}[section]

\theoremstyle{definition}
\newtheorem{definition}[theorem]{Definition}

\usepackage[numbers]{natbib}

\usepackage{mathtools}
\usepackage{xcolor} 

\usepackage{authblk}

\usepackage{subcaption}

\usepackage{mdframed}
\mdfsetup{linewidth=1pt}

\usepackage{xspace}
\usepackage{fancyhdr}
\usepackage{tcolorbox}

\newcommand{\bv}{\mathbf{v}}
\newcommand{\bmu}{\boldsymbol{\mu}}

\newcommand{\SW}{\text{\normalfont SW}}
\newcommand{\OPT}{\text{\normalfont OPT}}
\newcommand{\ALG}{\text{\normalfont ALG}}

\allowdisplaybreaks

\title{\bf The Price of EF1 for Few Agents with Additive Ternary Valuations}
\author{Maria Kyropoulou and Alexandros A. Voudouris}
\date{School of Computer Science and Electronic Engineering \\ University of Essex, Colchester, UK}

\begin{document}

\maketitle

\begin{abstract}
We consider a resource allocation problem with agents that have additive ternary valuations for a set of indivisible items, and bound the price of envy-free up to one item (EF1) allocations. For a large number $n$ of agents, we show a lower bound of $\Omega(\sqrt{n})$, implying that the price of EF1 is no better than when the agents have general subadditive valuations. We then focus on instances with few agents and show that the price of EF1 is $12/11$ for $n=2$, and between $1.2$ and $1.256$ for $n=3$.
\end{abstract}

\textbf{Keywords:} Fair division; EF1 allocations; Price of EF1

\section{Introduction}
The problem of fairly allocating scarce resources to agents with heterogeneous preferences has been a topic of interest since ancient times (e.g. in the form of land division) and has received significant attention within multiple disciplines in the last century; see \citep{survey,Moulin2019,Suksompong2021,Walsh2020} for some recent surveys. The main objective in this area is to compute or show the existence of fair allocations under assumptions about how fairness is defined exactly, the preferences of the agents, and the nature of the resources which might be divisible or indivisible, as well as desirable (goods) or undesirable (chores). In addition to fairness, another important desideratum is efficiency, i.e., to (approximately) optimize a function of the values or costs of the agents for the resources. Following a significant body of recent work, we consider a fair resource allocation setting with indivisible items, and aim to quantify the loss of efficiency in terms of social welfare (the total value of the agents for their assigned items) due to the requirement for fairness. 

\subsection{The General Model} \label{sec:model}
We consider a fair resource allocation setting with a set $N$ of $n \geq 2$ {\em agents} and a set $M$ of $m \geq 2$ {\em indivisible items}. Each agent $i$ has an {\em additive valuation function} $v_i:M \rightarrow \mathbb{R}_{\geq 0}$ such that $v_i(g) \geq 0$ is the value of agent $i$ for an item $g \in M$ (so, each item is considered a {\em good}), and $v_i(S) = \sum_{g \in S} v_i(g)$ is the value of $i$ for a subset of items $S \subseteq M$; also, each item is valued positively by at least one agent. So, an instance of the setting is fully described by the set of agents $N$, the set of items $M$, and the values of the agents for each individual items $\bv = (v_i)_{i \in N}$. 

An {\em allocation} $\mathbf{A} = (A_i)_{i \in N}$ is a collection of disjoint subsets (or, {\em bundles}) of $M$ so that agent $i$ is given the items in bundle $A_i$. Our objective is to compute an allocation that the agents consider to be fair. Many different notions of fairness have been proposed and studied in the fair division literature, with one of the most prominent ones being {\em envy-freeness}. 
This notion requires that no agent $i$ envies any other agent $j$ in the sense that the value of $i$ for $A_i$ is at least as much as the value of $i$ for $A_j$. However, it is not hard to observe that, when the items are indivisible, envy-free allocations may not exist; for example, if there are two agents and one item, then the agent that does not get the item will envy the other. In this work, we focus on probably the most well-studied relaxation of envy-freeness, known as {\em envy-freeness up to one item}, which was first introduced by \citet{Budish11} but was also informally studied before by \citet{LMMS04}. This notion allows an agent $i$ to envy another agent $j$ as long as this envy can be eliminated by (hypothetically) removing some item from $A_j$. 

\begin{definition}\label{def:EF1}
An allocation $\mathbf{A} = (A_i)_{i \in N}$ is {\em envy-free up to one item} (EF1) if, for any two agents $i$ and $j$ with $A_j \neq \varnothing$, there exists an item $g \in A_j$ such that $v_i(A_i) \geq v_i(A_j \setminus g)$. 
\end{definition}

EF1 allocations are known to always exist and can be computed via a variety of different methods. For example, the {\sc Round-Robin} algorithm allows the agents to choose their most valuable item one after the other in rounds determined by an (arbitrary) ordering of them. It is not hard to see that this algorithm terminates with an EF1 allocation: When an agent $i$ picks an item, she does so because she has more value for it than for any other item that agents picking after $i$ can choose. Hence, $i$ can only envy another agent $j$ because of a single item that $j$ picked before $i$ was able to pick any item in the very first round, and thus $i$'s envy can be eliminated by removing this item~\citep{Caragian2019MNW}. For other methods for computing EF1 allocations (possibly in combination with other constraints), we refer the reader to the survey of \citet{survey}.  

The {\em social welfare} $\SW(\mathbf{A})$ of an allocation $\mathbf{A}$ is the total value that the agents have for the sets of items they are allocated. Formally,  
\begin{align*}
    \SW(\mathbf{A}) = \sum_{i \in N} v_i(A_i). 
\end{align*}
To be able to meaningfully compare different allocations in terms of their social welfare, we assume that the valuation functions are normalized such that all agents have the same total value for all items: $v_i(M) = v_j(M)$ for any agents $i$ and $j$. 
With few exceptions, this is a standard normalization assumption in the relevant literature and is motivated, for example, by applications where all agents are given the same number of points (say 100) and are asked to distribute them among the items to indicate their preferences; see \citep{aziz19justifications} for further examples and justifications of this and many other normalization assumptions. It is also worth noting that, due to the nature of the EF1 criterion (which compares the value of each independent agent for different bundles of items), the set of possible EF1 allocations is not affected by our normalization assumption.

Unfortunately, EF1 allocations do not necessarily maximize the social welfare. To give an example, suppose there are two agents with values $v_1=(0.5,0.5,0)$ and $v_2=(0.49,0.26,0.25)$ for three items. The allocation that maximizes the social welfare assigns each item to the agent that has maximum value for it; that is, agent $1$ gets the first two items and agent $2$ gets the last item, resulting in a social welfare of $1.25$. However, observe that this allocation is not EF1: Agent $2$ has value $0.75$ for the first two items (which are given to agent $1$), and her value for each of the two items is strictly larger than the value she has for the last item that she is given. In any EF1 allocation, agent 2 has to be allocated one of the first two items that agent 1 values positively, and so the maximum social welfare in an EF1 allocation cannot be more than $1.24$. In fact, in general, it is NP-hard to compute the EF1 allocation with maximum social welfare~\citep{Bu2022complexity}. We are interested in quantifying how low the social welfare of any EF1 allocation can be compared to the optimal social welfare. This loss of efficiency is measured by the {\em price of EF1}, defined as the worst-case ratio (over all possible instances) between the maximum possible social welfare achieved by any allocation and the maximum possible social welfare achieved by an EF1 allocation. 
 
\begin{definition}\label{def:PoEF1}
Let $\texttt{EF1}(I)$ be the set of all possible EF1 allocations for instance $I=(N,M,(v_i)_{i \in N})$. 
The {\em price of EF1} is 
\begin{align*}
    \sup_{I} \frac{\max_{\mathbf{A}} \SW(\mathbf{A})}{\max_{\mathbf{A} \in \texttt{EF1}(I)} \SW(\mathbf{A})}.
\end{align*}
\end{definition}

The price of EF1 was first considered in the (conference version of the) work of \citet{Bei2021poef1}. 
For any number of agents $n$ with additive valuations, they showed that the price of EF1 is between $\Omega(\sqrt{n})$ and $O(n)$, with the upper bound achieved by {\sc Round-Robin}. This gap was resolved by \citet{Barman2020optimal} who designed an algorithm that computes an EF1 allocation with price $O(\sqrt{n})$, which holds even for subadditive valuation functions. \citeauthor{Bei2021poef1} also considered the special case of $n=2$ agents, for which they showed that the price of EF1 is between $8/7 \approx 1.143$ and $2/\sqrt{3} \approx 1.155$; a tight upper bound of $8/7$ for $n=2$ was recently shown by \citet{Li2024landscape}. 

\subsection{Our Contribution}
In this work, we revisit the price of EF1 allocations by focusing on instances in which the values of the agents for the items are {\em ternary} and can be of three different levels. In particular, the values of the agents may be $a$, $b$, or $0$ with $a > b > 0$. To give an example, consider an instance with two agents that have values $v_1=(2,1,1,0)$ and $v_2 = (1,1,1,1)$ for four items; here we have that $a=2$ and $b=1$, and also observe that the total value of both agents for all items is $4$. 
Such ternary values capture applications where agents can only partition the items into three different value categories to indicate their preferences. For instance, if the items correspond to houses, the categories could be houses the agents are eager to buy, houses they are willing to buy, and houses they do not want to buy at all. More generally, such instances can be thought of a special case of $k$-valued instances, where $k\geq 2$ values are available.

We first show that, even such restricted ternary valuations, the price of EF1 is $\Omega(\sqrt{n})$ when there is a large number of agents. This slightly improves the lower bound of $\Omega(\sqrt{n})$ shown by \citet{Bei2021poef1} in the sense that their result follows by an instance with four different values, whereas ours needs only three. Given that the algorithm of \citet{Barman2020optimal} computes an EF1 allocation for any class of subadditive valuation functions, our lower bound implies that no improvement on the price of EF1 is possible when there are many agents, even for ternary valuations. Hence, we next focus on instances with few agents, in particular, $n=2$ and $n=3$. 

For $n=2$ agents, we show that the price of EF1 is exactly $12/11$. The upper bound follows by a variation of the {\sc Round-Robin} algorithm that was discussed in Section~\ref{sec:model}. The main difference is that, when it is the turn of an agent to pick, she chooses an item among her most valuable ones by giving priority to those that do not harm the other agent. In a sense, our algorithm aims to minimize the social welfare loss as agents pick items greedily. For $n=3$ agents, we show that the price of EF1 is between $6/5=1.2$ and $1.256$. The upper bound in this case follows by an algorithm that constructs an allocation by repeatedly computing maximum matchings between the agents and the remaining items, again aiming to minimize the social welfare loss when there are multiple matchings that yield the same total value gain. 

\subsection{Other Related Work}
The price of EF1 is just an instantiation of a more general notion known as the {\em price of fairness}, which quantitatively compares the social welfare of allocations that satisfy particular fairness criteria to the maximum possible unrestricted social welfare. It was first introduced independently by \citet{Bertsimas2011pof} and \citet{Caragian2012efficiency}, who showed bounds on the price of envy-freeness and the price of proportionality for settings with divisible as well as indivisible items; for the latter, these results are only for classes of instances that admit such fair allocations (since envy-free and proportional allocations might not exist). 
Besides the price of EF1, \citet{Bei2021poef1}, \citet{Barman2020optimal} and \citet{Li2024landscape} also showed bounds on the price of many other fairness notions for instances with indivisible items, such as balancedness, Nash welfare, EFX for two agents, $1/2$-MMS, and EFM. Bounds on the price of fairness have also been shown for other criteria in settings with divisible items~\citep{Nicosia2017bounded},
chores~\citep{Aziz2024PropX,Heydrich2015chores,Wu2023WEF1}, and even for committee voting~\citep{Elkind2024representation}. 

Besides our paper, ternary values have also been explicitly considered in many other works within the fair division literature. For example, \citet{Amanatidis2017approximation} showed that exact MMS allocations exist and can be computed efficiently for the special case where  $a=2$ and $b=1$. On the negative side, \citet{fitzsimmons2024hardness} showed that approximately maximizing the Nash welfare or the egalitarian welfare remains NP-hard even for ternary values. \citet{Bhaskar2024trilean} recently considered instances with trilean values, where agents may have three different values for bundles of items (not necessarily singletons), and showed the existence of EF1 allocations. Instances with ternary values can also be thought of generalizations of bivalued instances  which have been studied significantly for notions like EFX and Maximum Nash welfare, see, for example~\citep{Amanatidis2021MNW,Babaioff2021dichotomous,Garg2023few,Halpern2020binary}.

\section{A Lower Bound for $n$ Agents}
We start by proving a lower bound of $\Omega(\sqrt{n})$ for when there is a large number $n$ of agents. As already mentioned, in combination with the upper bound of $O(\sqrt{n})$ of \citet{Barman2020optimal}, this lower bound shows that one cannot hope for asymptotically better price of EF1, even when the valuations are very simple and of just three levels. 

\begin{theorem}
For $n \geq 4$ agents with additive ternary valuations, the price of EF1 is $\Omega(\sqrt{n})$.
\end{theorem}

\begin{proof}
Consider an instance with $n$ agents and $n$ items. To define the values of the agents, we partition the items into a collection $S$ of $\sqrt{n}$ sets $S_1,\ldots,S_{\sqrt{n}}$ of size $\sqrt{n}$. Each set $S_i$ is associated with a different agent who has value $a$ for the items in $S_i$ and $0$ for any other item. Let $X$ be the set of agents associated with the sets of collection $S$. The remaining $n-\sqrt{n}$ agents have value $b$ for all items. Let $Y = N\setminus X$. Due to normalization, it must be the case that $a \cdot \sqrt{n} = b \cdot n \Leftrightarrow a = b \cdot \sqrt{n}$.

In any allocation where an agent $i \in Y$ is not assigned any items, there must exist another agent $j$ that is assigned at least two items, and hence $i$ is not EF1 towards $j$. This implies that each agent of $Y$ must be assigned at least one item for the allocation to be EF1. 
Hence, the maximum possible social welfare of an EF1 allocation is 
\begin{align*}
    b \cdot (n-\sqrt{n}) + a \cdot \sqrt{n} = b \cdot (2n - \sqrt{n})
\end{align*}
On the other hand, the optimal social welfare is $a \cdot n = b \cdot n\sqrt{n}$ and is achieved by assigning all the items to the agents that have value $a$ for them. Hence, the price of EF1 is 
\begin{align*}
    \frac{n\sqrt{n}}{2n-\sqrt{n}}. 
\end{align*}
It is not hard to observe that this expression is $\Omega(\sqrt{n})$, which completes the proof. 
\end{proof}

\section{Two Agents} \label{sec:two}
We start with the case of two agents for which we show a tight bound of $12/11$ on the price of EF1. 
We first show the lower bound. 

\begin{theorem}
For two agents with additive ternary valuations, the price of EF1 is at least $12/11$. 
\end{theorem}

\begin{proof}
Consider an instance with two agents and four items. The values of the agents for the items are $v_1 = (3/2,3/2,3/2,0)$ and $v_2 = (1,1,1,3/2)$. If we allocate the three first items to agent $1$, then agent $2$ will only get the last item and will not be EF1 towards agent $1$. So, at least one of the first three items must be given to agent $2$ in any EF1 allocation, leading to social welfare of at most $3\cdot 3/2 + 1 = 11/2$. The optimal social welfare is $4 \cdot 3/2 = 12/2$ (achieved by assigning each item to the agent that values it the most), leading to a price of EF1 of at least $12/11$.  
\end{proof}

For the upper bound, we consider a variation of the {\sc Round-Robin} algorithm to which we refer as {\sc Modified-2Agents-Round-Robin} (M2RR); see Algorithm~\ref{alg:M2RR}. The agents take turns choosing their most-valuable(positively-valued) item from the pool of remaining ones, breaking ties in favor of the item(s) that will minimally impair the social welfare. More formally, the agents choose among the items for which they have maximum (positive) value, the one for which the other agent has minimal value. For ternary valuations in particular, an agent prefers items she values as $a$ over items she values as $b$, and when she has to choose between multiple same-valued items, she prioritizes items that the other agent values as $0$, over the ones the other agent values as $b$, over the ones the other agent values as $a$. If an agent has value $0$ for all remaining items, all these items are given to the other agent. 

\SetCommentSty{mycommfont}
\begin{algorithm}[ht]
\SetNoFillComment
\caption{\sc Modified-2Agents-Round-Robin (M2RR)}
\label{alg:M2RR}
Order the agents such that agent $1$ is the one with the most $a$s\;
\For{each $i \in [2]$}
{
    $A_i \gets \varnothing$\;
}
$P \gets M$\;
\While{ $P \neq \varnothing$}
{
    \For{each $i \in [2]$}
    {
        \If{ $v_i(x) = 0$ for each $x \in P$}
        {
            $A_{3-i} \gets A_{3-i} \cup P$\;
            $P \gets \varnothing$\; 
        }
        \Else 
        {
            $C \gets \arg\max_{x \in P}\{v_i(x)\}$\;
            $x_i \gets $ arbitrary item in $\arg\min_{x \in C}\{v_{3-j}(x)\}$\;
            $A_i \gets A_i \cup \{x_i\}$\;
            $P \gets P \setminus \{x_i\}$\;
        }
    }
}
\Return $\mathbf{A} = (A_1, A_2)$\; 
\end{algorithm}

\begin{theorem}
For two agents with additive ternary valuations, 
M2RR computes an allocation that is EF1 and has price of EF1 at most $12/11$. 
\end{theorem}

\begin{proof}
The allocation computed by M2RR is EF1 since {\sc Round-Robin} always outputs an EF1 allocation, and if an agent stops receiving items, it means that her value is $0$ for each of the remaining ones, and thus does not have any envy for them. We next focus on bounding the price of EF1. 

Fix the values $a$ and $b$ that the agents might have for the items. To describe the universe of all possible instances with $a$ and $b$ fixed, let $S_{x,y}$ be the set of items for which agent $1$ has value $x \in \{a,b,0\}$ and agent $2$ has value $y \in \{a,b,0\}$. In addition, let $\ALG(x,y)$ and $\OPT(x,y)$ be the contribution to the social welfare of items in $S_{x,y}$ in the computed EF1 allocation and in the optimal allocation, respectively.
Observe that for valid pairs $(x,y)$ with $x=y$ or $x=0$ or $y=0$, it holds that $\ALG(x,y)=\OPT(x,y)$. Using this, since the price of EF1 is the ratio of the optimal social welfare over the social welfare of the computed EF1 allocation, we can assume that, in a worst-case instance (where the price of EF1 is maximized), there is a minimal number of such items. In particular, since some agent must have value $0$ for some item (as otherwise the price of EF1 would be $1$), we have that 
$S_{a,a} \cup S_{b,b} = \varnothing$ 
and $S_{a,0} \cup S_{0,a} \cup S_{b,0} \cup S_{0,b} = \{g\}$.

Any inefficiency in social welfare comes from incorrectly allocating items in $S_{a,b} \cup S_{b,a}$. For such items, observe that, as long as there are available items in both sets, since the agents aim to minimize the value loss of the other agent when choosing items, agent $1$ picks from $S_{a,b}$ and agent $2$ picks from $S_{b,a}$ leading, again, to the same contribution in the social welfare as in the optimal solution (since both agents pick $a$-valued items). So, the loss of efficiency is due to one of these sets being empty. 
By the definition of the algorithm, since agent $1$ is the one with more $a$s in her valuation, it must be the case that $S_{b,a} = \varnothing$, and hence, the agents pick from the items in $S_{a,b}$; let $k := |S_{a,b}|$. Since $a > b$, it must be the case that $g \in S_{0,a} \cup S_{0,b}$ (as otherwise, agent $1$ would have a total value either $(k+1)\cdot a$ if $S_{a,0} \neq \varnothing$ or $k\cdot a + b$ if $S_{b,0} \neq \varnothing$, while agent $2$ would only have value $k\cdot b$, breaking the normalization assumption).  

Now that we have the structure of a worst-case instance, we are ready to bound the price of EF1. 
By the definition of the algorithm, since agent $1$ values $g$ as $0$, agent $1$ gets an item from $S_{a,b}$, agent $2$ gets $g$, and then they pick the remaining $k-1$ items of $S_{a,b}$ one by one with agent $1$ picking first (to lead to an EF1 allocation with best possible social welfare). 
Let $x$ and $y$ be the number of the $k-1$ items in $S_{a,b}$ that the two agents pick after the first round, respectively.
Then, the social welfare of the allocation computed by the algorithm is $a + v_2(g) + x\cdot a + y \cdot b$, while the optimal social welfare is $k\cdot a + v_2(g)$. We distinguish between the possible values of $v_2(g)$. 

\medskip
\noindent 
{\bf Case 1: $v_2(g) = b$.}
By the normalization, both agents must have the same total value for all items, and thus
$k \cdot a = (k+1)\cdot b \Leftrightarrow a = \frac{k+1}{k}\cdot b$. Since $x+y = k-1$, $x \geq \frac{k-1}{2}\geq y$, and $k \geq 3$ (since otherwise the $k$ items in $S_{a,b}$ and $g$ would all be allocated optimally by the algorithm), the price of EF1 is 
\begin{align*}
    \frac{k\cdot a + b }{(x+1)\cdot a +(y+1)\cdot b}
    &= \frac{k \cdot \frac{k+1}{k} + 1}{(x+1)\cdot \frac{k+1}{k} + (y+1)} \\
    &= \frac{k(k+2)}{k\cdot(x+y) + x+2k+1} \\
    &\leq \frac{2k^2 + 4k}{2k^2 +3k+1} \\
    &\leq \frac{15}{14}.
\end{align*}

\medskip
\noindent 
{\bf Case 2: $v_2(g) = a$.}
By the normalization, we now have that $k \cdot a = k\cdot b + y \Leftrightarrow a = \frac{k}{k-1}\cdot b$. Since $x+y = k-1$, $x \geq \frac{k-1}{2}\geq y$, and $k \geq 3$, the price of EF1 is 
\begin{align*}
    \frac{(k+1)\cdot a }{(x+2)\cdot a +y\cdot b } 
    &= \frac{(k+1) \cdot \frac{k}{k-1}}{ (x+2) \cdot \frac{k}{k-1} + y} \\
    &= \frac{k(k+1)}{k(x+y) + 2k -y} \\
    &\leq \frac{2k^2+2k}{2k^2+k+1} \\
    &\leq \frac{12}{11}.
\end{align*}
This concludes the proof.
\end{proof}

\section{Three Agents}\label{sec:three}
For three agents, we show that the price of EF1 is between $1.2$ and $1.256$. We again start with the lower bound. 

\begin{theorem}
For three agents with additive ternary valuations, the price of EF1 is at least $6/5$.
\end{theorem}

\begin{proof}
Consider an instance with three agents and six items. The values of the agents for the items are
$v_1 = (2,2,2,0,0,0)$,
$v_2 = (0,0,0,2,2,2)$,
and 
$v_3 = (1,1,1,1,1,1)$.
Clearly, in any EF1 allocation, all agents must take at least one item. Also, observe that if we allocate just one item to agent $3$ and all three items with value $2$ to agent $1$ or $2$, then $3$ will not be EF1 towards that agent. Hence, in an EF1 allocation, agent $3$ must take at least two items, and the social welfare is at most $4\cdot 2 + 2 = 10$. The optimal social welfare is $6\cdot 2 = 12$ (achieved by assigning the items to the first two agents), leading to a price of EF1 of at least $6/5$. 
\end{proof}

For the upper bound, we consider an algorithm to which we refer as {\sc Repeated-Max-Matching} (RMM); see Algorithm~\ref{alg:RMM}. 
The algorithm constructs an allocation in rounds by repeatedly computing max matchings on a bipartite graph with the nodes on one side representing the agents and the nodes on the other side representing the items. An edge between an agent and an item exists in this graph if and only if the corresponding value of the agent is positive (either $a$ or $b$), and has weight equal to this value. Each time a maximum matching is computed, the graph is updated so that the items in the matching and all of their adjacent edges are removed. In addition, if an agent has no adjacent edges left in the graph, she is also removed as she has value $0$ for all of the remaining items. 

In case there are multiple max matchings, the algorithm chooses a {\em non-wasteful} one in the sense that each agent is matched to an  item that minimizes the maximum value that the other two agents have for it. To illustrate this, suppose that agent $i$ is matched to an item $g$ which $i$ and another agent $j$ both value as $a$. Then, there must be no unmatched item $q$ such that $i$ values as $a$ and both other agents value as $b$ or $0$; otherwise, agent $i$ would have been matched to $q$ instead, possibly allowing $g$ to be matched to $j$. 

\SetCommentSty{mycommfont}
\begin{algorithm}[ht]
\SetNoFillComment
\caption{\sc Repeated-Max-Matching (RMM)}
\label{alg:RMM}
\For{each $i \in N$}
{
    $A_i \gets \varnothing$\;
}

Construct bipartite graph $G=(N,M,E,\bv)$ with $(i,g) \in E$ iff $v_i(g) > 0$\;

\While{ $M \neq \varnothing$}
{
    Compute a non-wasteful matching $\bmu = (\mu_i)_{i \in N}$ on $G$\;
    \For{each $i \in N$}
    {
        $A_i \gets A_i \cup \{\mu_i\}$\;
        $M \gets M \setminus \{\mu_i\}$\;
        $E \gets E \setminus \{e\in E: \mu_i \in e\}$\;
        \If{$\nexists e\in E: i \in e$}
        {
            $N \gets N \setminus \{i\}$\;
        }
    }
}
\Return $\mathbf{A} = (A_i)_{i \in N}$\; 
\end{algorithm}

\begin{theorem}
For three agents with additive ternary valuations, RMM computes an allocation that is EF1 and has price of EF1 at most $24/19 \approx 1.26$.
\end{theorem}

\begin{proof} 
We first argue that the allocation computed by the algorithm is EF1. Consider an arbitrary agent $i \in [3]$ and the items she is allocated as the algorithm computes max matchings in rounds. Suppose that agent $i$ is assigned items that she values as $a$ for the first $k-1$ rounds before being given an item $g$ that she values as $b$ in round $k \geq 1$; note that if $i$ receives only $a$-valued items and then she has $0$ for the remaining ones, then $i$ is envy-free (not just EF1) towards the other two agents.
It might happen that, in round $k$, one or both of the other two agents are allocated items that $i$ values as $a$ (which $i$ strictly prefers to the item she got in round $k$.)
Afterwards, suppose that agent $i$ is assigned items that she values as $b$ for the next $\lambda \geq 1$ rounds and then stops receiving items from round $k+\lambda+1$ and on. Now, it might happen that, in round $k+\lambda+1$, one or both of the other two agents are allocated items that $i$ values as $b$,
but from round $k+\lambda+2$ (if such rounds exist), agent $i$ definitely has value $0$ for the remaining items. Hence, $i$ has accumulated a total value of 
$(k-1)\cdot a + b + \lambda \cdot b = (k-1)\cdot a +(\lambda+1)\cdot b$ 
and might have a total value of at most 
$(k-1)\cdot a  + a + \lambda\cdot b + b = k \cdot a + (\lambda+1) \cdot b$ for the bundles of items of the other agents. By removing one of the $a$-valued items, we have that $i$ is EF1 towards the other agents, showing that the allocation is EF1.  

We next focus on proving the upper bound on the price of EF1. It will be helpful to partition the max matchings computed by the algorithm into a set of phases depending on the values of the agents for their assigned items. We remark that these phases do not consist of necessarily consecutive matchings computed as the algorithms runs.
\begin{itemize}
\item
Phase $1$: All agents receive items they value as $a$. 

\item
Phase $2$: Each item allocated in the matchings of this phase is valued $a$ by at least one agent. Note that no three of these items can be allocated to three different agents that value them as $a$, as then such a set of items would be part of Phase 1. 

\item
Phase $3$: All agents receive items they value as $b$, with the possible exception of the first matching in which there might be an item that some agent values as $a$, but at least one of the items is valued at most $b$ by all agents. 
\end{itemize}
For $p \in [3]$, let $M_p$ be the set of items in phase $p$. Denote by $\ALG(M_p)$ and $\OPT(M_p)$ the contribution of the items in $M_p$ in the social welfare of the computed EF1 allocation and of the optimal allocation, respectively. 

First observe that the social welfare contribution of the items in $M_1$ and $M_3$ is at least as much as their contribution in the optimal social welfare. Indeed, the value gain from each item in Phase $1$ as well as at most two items in the first matching of Phase $3$ is $a$, which is the maximum possible. The remaining items in Phase $3$ are allocated to agents that value them as $b$, but we also know that all agents have value at most $b$ for these items; hence, the contribution is again the maximum possible. So, we have that 
$\ALG(M_1) + \ALG(M_3) \geq \OPT(M_1) + \OPT(M_3)$; given this, we can thus assume that Phases $1$ and $3$ are empty in a worst-case instance.  

We now focus on the allocation of Phase $2$. 
Without loss of generality, suppose that agents $1$ and $2$ are the two agents with the largest and second largest numbers of $a$s, respectively, in their valuations for the items in $M_2$. In particular, let $x$ be the number of items that agent $1$ values as $a$, and let $y$ be the number of items that agent $2$ values as $a$ but agent $1$ values as at most $b$; by definition, $x \geq y$. We also have that $|M_2| = x+y$, as, otherwise, if there was an additional item which agent $3$ values as $a$ but the other two agents value as at most $b$, then there would be an additional matching that could be included in Phase $1$. Hence, since for each item in $M_2$ there is an agent that values it as $a$, we have that $\OPT(M_2) = (x+y)\cdot a$.

For the social welfare of the algorithm, observe that each of the first $y/2$ matchings yields a gain of $2a+b$ since both agents $1$ and $2$ obtain items they value as $a$. 
Afterwards, agent $2$ has no other item that she values as $a$, and hence the remaining $x-y/2$ items (all of which are valued as $a$ by agent $1$) might be allocated as follows: 
There is a number $z$ of matchings each of which allocates an $a$-valued item to agent $1$ and $b$-valued items to agents $2$ and $3$ for a gain of $a+2b$; since there can be at most $(x-y/2)/3$ such matchings, we have that $0 \leq z \leq (2x-y)/6$. Following these matchings, agent $2$ is left with only items that she values as $0$ and is thus removed. Hence, the remaining $x-y/2 - 3z$ items are then allocated to agents $1$ and $3$ during the final $(x-y/2-3z)/2$ matchings, each of which yields a gain of at least $a+b$. Consequently,
\begin{align*}
    \ALG(M_2) 
    &\geq \frac{y}{2}  \cdot (2a+b)+ z \cdot (a+2b) + \left(\frac{x- y/2 - 3z}{2}\right) (a+b) \\
    &= \frac{y}{2}  \cdot (2a+b) + \frac{x- y/2}{2} \cdot  (a+b) - z \cdot \frac{a-b}{2} \\
    &\geq \frac{y}{2}  \cdot (2a+b)+ \frac{2x- y}{6} \cdot (a+2b),
\end{align*}
since $z \leq (2x-y)/6$. 
Putting everything together, the price of EF1 can be upper-bounded as follows:
\begin{align}
\frac{\OPT(M_1)+\OPT(M_2)+\OPT(M_3}{\ALG(M_1)+\ALG(M_2)+\ALG(M_3)} 
&\leq \frac{\OPT(M_2)}{\ALG(M_2)} \nonumber \\
&\leq \frac{(x+y) a}{\frac{y}{2}(2a+b)+ \frac{2x- y}{6}\left(a+2b\right)} \nonumber \\
&=\frac{(1+\frac{y}{x}) a}{\frac{1}{2}\frac{y}{x}(2a+b)+ \left(\frac{1}{3}-\frac{1}{6}\frac{y}{x}\right)\left(a+2b\right)} \label{function}
\end{align}
Let $w=y/x \in [0,1]$. 
By the normalization of the valuations and the fact that agent 3 has value at most $(x+y)b$ for the items in phase $M_2$, we obtain the following two constraints:
\begin{itemize}
\item Since agent 1 has value $a$ for $x$ items, we have $x\cdot a\leq (x+y) \cdot b \Leftrightarrow a\leq (1+w)b$.
\item Since agent 2 has value $a$ for at least $y$ items and value $b$ for at least $z$ items, we have $y \cdot a + z \cdot b \leq (x+y) b \Leftrightarrow a\leq \frac{4+7w}{6w} b$, where the last inequality follows by the assumption that $z$ takes its maximum value of $(2x-y)/6$ (which minimizes the social welfare of the algorithm). 
\end{itemize}
Given that $a$ and $b$ must satisfy both inequalities, we can combine them into a single one:
\begin{align*}
    a \leq \min \left\{ 1+w, \frac{4+7w}{6w} \right\} \cdot b.
\end{align*}
Note that the minimum evaluates to $1+w$ when $w \leq \frac{1}{12}(1+\sqrt{97}) \approx 0.904$, and $\frac{4+7w}{6w}$ otherwise; it is worth noting that $w$ should actually be a rational number as the ratio of two integers $x$ and $y$, but, as we aim for an upper bound on the price of EF1, the irrational upper bound of $w$ is sufficient (although it does not lead to a tight analysis). Since \eqref{function} is increasing in $a$, we can upper-bound the price of EF1 by substituting $a$ with the right-hand side of the above inequality. In particular:
\begin{itemize}
\item If $w \leq \frac{1}{12}(1+\sqrt{97})$, we obtain 
\begin{align*}
\frac{(1+w)^2}{\frac{1}{2}w(3+2w)+ \left(\frac{1}{3}-\frac{1}{6}w\right)(3+w)}
= \frac{6(1+w)^2}{5w^2+8w+6}.
\end{align*}
This expression is an increasing function of $w$ in the interval $[0,\frac{1}{12}(1+\sqrt{97})]$ and has a maximum value of 
$\frac{3}{586}(137+11\sqrt{97}) \approx 1.256$. 

\item If $w \geq \frac{1}{12}(1+\sqrt{97})$, we obtain 
\begin{align*}
\frac{(1+w)\cdot \frac{4+7w}{6w}}{ \frac12 w \left(2 \frac{4+7w}{6w} +1 \right) + (\frac13 - \frac16 w) \left( \frac{4+7w}{6w} +2 \right)}
= \frac{6(1+w)(7w+4)}{41w^2+58w+8}.
\end{align*}
This expression is a decreasing function of $w$ in the interval $[\frac{1}{12}(1+\sqrt{97}),1]$ and has a maximum value of 
$\frac{3}{586}(137+11\sqrt{97}) \approx 1.256$, again. 
\end{itemize}
The proof is now complete. 
\end{proof}

\bibliographystyle{plainnat}
\bibliography{references}

\end{document}